\documentclass[a4paper,UKenglish,cleveref, autoref]{lipics-v2019}

\usepackage{tikz}
\usepackage[UKenglish]{babel}

\newtheorem{observation}{Observation}

\newtheorem{question}{Question}

\title{Hardness and approximation for the geodetic set problem in some graph classes} 

\author{Dibyayan Chakraborty}{Indian Statistical Institute, Kolkata, India.}{dibyayancg@gmail.com}{}{}

\author{Florent Foucaud}{Univ. Orl\'eans, INSA Centre Val de Loire, LIFO EA 4022, F-45067 Orl\'eans Cedex 2, France.} {florent.foucaud@gmail.com}{}{Partially supported by the ANR project HOSIGRA (ANR-17-CE40-0022).}

\author{Harmender Gahlawat}{Indian Statistical Institute, Kolkata, India.} {harmendergahlawat@gmail.com}{}{}

\author{Subir Kumar Ghosh}{Ramakrishna Mission Vivekananda Educational and Research Institute.}{subir.ghosh@rkmvu.ac.in}{}{}

\author{Bodhayan Roy}{Indian Institute of Technology, Kharagpur.}{bodhayan.roy@gmail.com}{}{}

\authorrunning{Chakraborty et al.}

\Copyright{Dibyayan Chakraborty,Florent Foucaud, Harmender Gahlawat, Subir Kumar Ghosh, Bodhayan Roy}

\keywords{Planar graph, Geodetic Set, Approximation}

\acknowledgements{This research is funded by the IFCAM project ``Applications of graph homomorphisms'' (MA/IFCAM/18/39). We would like to thank Ajit A. Diwan for the many helpful pointers that he provided.}

\EventEditors{John Q. Open and Joan R. Access}
\EventNoEds{2}
\EventLongTitle{42nd Conference on Very Important Topics (CVIT 2016)}
\EventShortTitle{CVIT 2016}
\EventAcronym{CVIT}
\EventYear{2016}
\EventDate{December 24--27, 2016}
\EventLocation{Little Whinging, United Kingdom}
\EventLogo{}
\SeriesVolume{42}
\ArticleNo{23}

\ccsdesc[500]{Theory of computation~Graph algorithms analysis.}

\begin{document}

\maketitle

\begin{abstract}
In this paper, we study the computational complexity of finding the \emph{geodetic number} of graphs. A set of vertices $S$ of a graph $G$ is a \emph{geodetic set} if any vertex of $G$ lies in some shortest path between some pair of vertices from $S$. The \textsc{Minimum Geodetic Set (MGS)} problem is to find a geodetic set with minimum cardinality. In this paper, we prove that solving the \textsc{MGS} problem is NP-hard on planar graphs with a maximum degree six and line graphs. We also show that unless $P=NP$, there is no polynomial time algorithm to solve the \textsc{MGS} problem with sublogarithmic approximation factor (in terms of the number of vertices) even on graphs with diameter $2$. On the positive side, we give an $O\left(\sqrt[3]{n}\log n\right)$-approximation algorithm for the \textsc{MGS} problem on general graphs of order $n$. We also give a $3$-approximation algorithm for the \textsc{MGS} problem on the family of solid grid graphs which is a subclass of planar graphs.
\end{abstract}
\section{Introduction and results}\label{sec:intro}

Suppose there is a city-road network (i.e. a graph) and a bus company wants to open bus terminals in some of the cities. The buses will go from one bus terminal to another (i.e. from one city to another) following the shortest route in the network. Finding the minimum number of bus terminals required so that any city belongs to some shortest route between some pair of bus terminals is equivalent to finding the \emph{geodetic number} of the corresponding graph. Formally, an undirected simple graph $G$ has vertex set $V(G)$ and edge set $E(G)$. For two vertices $u,v\in V(G)$, let $I(u,v)$ denote the set of all vertices in $G$ that lie in some shortest path between $u$ and $v$. A set of vertices $S$ is a \emph{geodetic set} if $\cup_{u,v\in S} I(u,v)=V(G)$. The \emph{geodetic number}, denoted as $g(G)$, is the minimum integer $k$ such that $G$ has a geodetic set of cardinality $k$. Given a graph $G$, the \textsc{Minimum Geodetic Set (MGS)} problem is to compute a geodetic set of $G$ with minimum cardinality. In this paper, we shall study the computational complexity of the \textsc{MGS} problem in various graph classes.

The notion of geodetic sets and geodetic number was introduced by Harary et al.~\cite{harary1993}. The notion of geodetic number is closely related to convexity and convex hulls in graphs, which have applications in game theory, facility location, information retrieval, distributed computing and communication networks~\cite{buckley1985,haynes2003,gerstel1994,mitchell1978,ekim2014}. In 2002, Atici~\cite{atici2002} proved that finding the geodetic number of arbitrary graphs is NP-hard. Later, Dourado et al.~\cite{dourado2008,dourado2010} strengthened the above result to \emph{bipartite} graphs, \emph{chordal} graphs and \emph{chordal bipartite} graphs. Recently, Bueno et al.~\cite{bueno2018} proved that the \textsc{MGS} problem remains NP-hard even for \emph{subcubic} graphs. On the positive side, polynomial time algorithms to solve the \textsc{MGS} problem are known for \emph{cographs}~\cite{dourado2010}, \emph{split} graphs~\cite{dourado2010}, \emph{ptolemaic} graphs~\cite{farber1986}, \emph{outer planar} graphs~\cite{mezzini2018} and \emph{proper interval} graphs~\cite{ekim2012}. In this paper, we prove the following theorem.

\begin{theorem}\label{thm:planar}
	The \textsc{MGS} problem is NP-hard for planar graphs of maximum degree $6$.
\end{theorem}




Then we focus on \emph{line} graphs. Given a graph $G$, the \emph{line} graph of $G$, denoted by $L(G)$ is a graph such that each vertex of $L(G)$ represents an edge of $G$ and two vertices of $L(G)$ are adjacent if and only if their corresponding edges share a common endpoint in $G$. A graph $H$ is a \emph{line graph} if $H\cong L(G)$ for some $G$. Some optimisation problems which are difficult to solve in general graphs admit polynomial time algorithms when the input is a line graph~\cite{gerber2003,guruswami1999}. We prove the following theorem.

\begin{theorem}\label{thm:line-graph}
	The \textsc{MGS} problem is NP-hard for line graphs.
\end{theorem}


From a result of Dourado et al.~\cite{dourado2010}, it follows that solving the \textsc{MGS} problem is NP-hard even for graphs with diameter at most $4$. On the other hand, solving the \textsc{MGS} problem on graphs with diameter $1$ is trivial (since those are exactly complete graphs). In this paper, we prove that unless P=NP, there is no polynomial time algorithm with sublogarithmic approximation factor for the \textsc{MGS} problem even on graphs with diameter at most $2$. A \emph{universal vertex} of a graph is adjacent to all other vertices of the graph. We shall prove the following stronger theorem. 

\begin{theorem}\label{thm:reduction-diam2}
	Unless P=NP, there is no polynomial time $o(\log n)$-approximation algorithm for the \textsc{MGS} problem even on graphs that have a universal vertex, where $n$ is the number of vertices in the input graph.
\end{theorem}


On the positive side, we show that a reduction to the \textsc{Minimum  Rainbow Subgraph of Multigraph} problem (defined in Section~\ref{sec:approx-general}) gives the first sublinear approximation algorithm for the \textsc{MGS} problem on general graphs. 

\begin{theorem}\label{thm:general-approx}
	Given a graph, there is a polynomial-time $O(\sqrt[3]{n}\log n)$-approximation algorithm for the \textsc{MGS} problem where $n$ is the number of vertices.
\end{theorem}



Then we focus on \emph{solid grid} graphs, an interesting subclass of planar graphs. A \emph{grid embedding} of a graph is a collection of points with integer coordinates such that each point in the collection represents a vertex of the graph and two points are at a distance one if and only if the vertices they represent are adjacent in the graph. A graph is a \emph{grid} graph if it has a grid embedding. A  graph is a \emph{solid grid} graph if it has a grid embedding such that all interior faces have unit area. Approximation algorithms for optimisation problems like \textsc{Longest path, Longest Cycle, Node-Disjoint Path} etc. on grid graphs and solid grid graphs have been studied~\cite{cualinescu2008,itai1982,chuzhoy2015,wu20107,sardroud2016,zhang2011}. In this paper, we prove the following theorem.

\begin{theorem}\label{thm:grid-approx}
	Given a solid grid graph, there is an $O(n)$ time $3$-approximation algorithm for the \textsc{MGS} problem, even if the grid embedding is not given as part of the input. Here $n$ is the number of vertices in the input graph.
\end{theorem}

Note that recognising solid grid graphs is NP-complete~\cite{gregori1989}.

\medskip\noindent\textbf{Organisation of the paper:} In Section~\ref{sec:hard}, we prove the hardness results for planar graphs, line graphs and graphs with diameter $2$. In Section~\ref{sec:approx}, we present our approximation algorithms. Finally we draw our conclusions in Section~\ref{sec:conclude}.

\section{Hardness results}\label{sec:hard}

In Section~\ref{sec:planar-hard}, we prove that the \textsc{MGS} problem is NP-hard for planar graphs with maximum degree $6$ (Theorem~\ref{thm:planar}). Then in Section~\ref{sec:line-hard} we prove that the \textsc{MGS} problem is NP-hard for line graphs (Theorem~\ref{thm:line-graph}). In Section~\ref{sec:approx-hard} we prove the inapproximability result (Theorem~\ref{thm:reduction-diam2}).

\subsection{NP-hardness on planar graphs}\label{sec:planar-hard}

Given a graph $G$, a subset $S\subseteq V(G)$ is a dominating set of $G$ if any vertex in $V(G)\setminus S$ has a neighbour in $S$. The problem \textsc{Minimum Dominating set (MDS)} consists in computing a dominating set of an input graph $G$ with minimum cardinality. To prove Theorem~\ref{thm:planar}, we shall reduce the NP-complete \textsc{MDS} problem on subcubic planar graphs~\cite{garey2002} to the \textsc{MGS} problem on planar graphs with maximum degree $6$.

Let us describe the reduction. From a subcubic planar graph $G$ with a given planar embedding, we construct a graph $f(G)$ as follows. Each vertex $v$ of $G$ will be replaced by a \emph{vertex gadget} $G_v$. This vertex gadget has vertex set $\{c^v,t^v_{0},t^v_{1},t^v_{2}\}\cup\{x_{i,j}^v,y_{i,j}^v,z_{i,j}^v~|~0\leq i<j\leq 2\}$. There are no other vertices in $f(G)$. For the edges within $G_v$, vertex $t_i^v$ (for $0\leq i\leq 2$) is adjacent to vertices $c^v$, $x_{i,i+1}^v$, $y_{i,i+1}^v$, $x_{i-1,i}^v$, $y_{i-1,i}^v$ (indices taken modulo~$3$). Moreover, for each pair $i,j$ with $0\leq i<j\leq 2$, $x_{i,j}^v$ is adjacent to $c^v$ and $y_{i,j}^v$, and $y_{i,j}^v$ is adjacent to $z_{i,j}^v$. We now describe the edges outside of the vertex-gadgets. They will depend on the embedding of $G$. We assume that the edges incident with any vertex $v$ are labeled $e_i^v$ with $0\leq i<deg_G(v)$, in such a way that the numbering increases counterclockwise around $v$ with respect to the embedding (thus the edge $vw$ will have two labels: $e_i^v$ and $e_j^w$). Consider two vertices $v$ and $w$ that are adjacent in $G$, and let $e_i^v$ and $e_j^w$ be the two labels of edge $vw$ in $G$. Then, $t_i^v$ is adjacent to $t_j^w$, $y_{i,i+1}^v$ is adjacent to $y_{j-1,j}^w$ and $y_{i-1,i}^v$ is adjacent to $y_{j+1,j}^w$ (indices are taken modulo the degree of the original vertex of $G$). It is clear that a planar embedding of $f(G)$ can easily be obtained from the planar embedding of $G$. Thus $f(G)$ is planar and has maximum degree $6$. The construction is depicted in Figure~\ref{fig:planar-reduction}, where $v$ and $w$ are adjacent in $G$ and the edge $vw$ is labeled $e_0^v$ and $e_0^w$.

\begin{figure}[!htpb]
	\centering
	\scalebox{0.7}{\begin{tikzpicture}[scale=1]
		\node[circle, draw=black!100,fill=black!100,thick, inner sep=0pt, minimum size=1.5mm,label=above right:{$c^v$}] (c^v) at (0,0) {};
		\node[circle, draw=black!100,fill=black!100,thick, inner sep=0pt, minimum size=1.5mm,label=above right:{$t_0^v$}] (t_0^v) at (0:2.5) {};
		\node[circle, draw=black!100,fill=black!100,thick, inner sep=0pt, minimum size=1.5mm,label=left:{$x_{0,1}^v$}] (x_{0,1}^v) at (60:1) {};
		\node[circle, draw=black!100,fill=black!100,thick, inner sep=0pt, minimum size=1.5mm,label=above left:{$y_{0,1}^v$}] (y_{0,1}^v) at (60:2) {};
		\node[circle, draw=black!100,fill=black!100,thick, inner sep=0pt, minimum size=1.5mm,label=right:{$z_{0,1}^v$}] (z_{0,1}^v) at (60:3) {};
		\node[circle, draw=black!100,fill=black!100,thick, inner sep=0pt, minimum size=1.5mm,label=left:{$t_1^v$}] (t_1^v) at (120:2.5) {};
		\node[circle, draw=black!100,fill=black!100,thick, inner sep=0pt, minimum size=1.5mm,label=above left:{$x_{1,2}^v$}] (x_{1,2}^v) at (180:1) {};
		\node[circle, draw=black!100,fill=black!100,thick, inner sep=0pt, minimum size=1.5mm,label=above left:{$y_{1,2}^v$}] (y_{1,2}^v) at (180:2) {};
		\node[circle, draw=black!100,fill=black!100,thick, inner sep=0pt, minimum size=1.5mm,label=below:{$z_{1,2}^v$}] (z_{1,2}^v) at (180:3) {};
		\node[circle, draw=black!100,fill=black!100,thick, inner sep=0pt, minimum size=1.5mm,label=left:{$t_2^v$}] (t_2^v) at (240:2.5) {};
		\node[circle, draw=black!100,fill=black!100,thick, inner sep=0pt, minimum size=1.5mm,label=left:{$x_{0,2}^v$}] (x_{0,2}^v) at (300:1) {};
		\node[circle, draw=black!100,fill=black!100,thick, inner sep=0pt, minimum size=1.5mm,label=below left:{$y_{0,2}^v$}] (y_{0,2}^v) at (300:2) {};
		\node[circle, draw=black!100,fill=black!100,thick, inner sep=0pt, minimum size=1.5mm,label=right:{$z_{0,2}^v$}] (z_{0,2}^v) at (300:3) {};
		\draw (t_0^v)--(c^v)--(x_{0,1}^v)--(y_{0,1}^v)--(z_{0,1}^v);
		\draw (t_1^v)--(c^v)--(x_{1,2}^v)--(y_{1,2}^v)--(z_{1,2}^v);
		\draw (t_2^v)--(c^v)--(x_{0,2}^v)--(y_{0,2}^v)--(z_{0,2}^v);
		\draw (t_0^v)--(x_{0,1}^v)--(t_1^v)--(x_{1,2}^v)--(t_2^v)--(x_{0,2}^v)--(t_0^v)--(y_{0,1}^v)--(t_1^v)--(y_{1,2}^v)--(t_2^v)--(y_{0,2}^v)--(t_0^v);
		\draw[dashed] (y_{1,2}^v) -- +(-1,-2);
		\draw[dashed] (y_{1,2}^v) -- +(-1,2);
		\draw[dashed] (y_{0,2}^v) -- +(-0.5,-1);
		\draw[dashed] (t_2^v) -- +(-0.25,-0.5);
		\draw[dashed] (y_{0,1}^v) -- +(-0.5,1);
		\draw[dashed] (t_1^v) -- +(-0.25,0.5);
		
		\node[xshift=8cm,circle, draw=black!100,fill=black!100,thick, inner sep=0pt, minimum size=1.5mm,label=above left:{$c^w$}] (c^w) at (0,0) {};
		\node[xshift=8cm,circle, draw=black!100,fill=black!100,thick, inner sep=0pt, minimum size=1.5mm,label=above left:{$t_0^w$}] (t_0^w) at (180:2.5) {};
		\node[xshift=8cm,circle, draw=black!100,fill=black!100,thick, inner sep=0pt, minimum size=1.5mm,label=right:{$x_{0,1}^w$}] (x_{0,1}^w) at (240:1) {};
		\node[xshift=8cm,circle, draw=black!100,fill=black!100,thick, inner sep=0pt, minimum size=1.5mm,label=below right:{$y_{0,1}^w$}] (y_{0,1}^w) at (240:2) {};
		\node[xshift=8cm,circle, draw=black!100,fill=black!100,thick, inner sep=0pt, minimum size=1.5mm,label=left:{$z_{0,1}^w$}] (z_{0,1}^w) at (240:3) {};
		\node[xshift=8cm,circle, draw=black!100,fill=black!100,thick, inner sep=0pt, minimum size=1.5mm,label=right:{$t_1^w$}] (t_1^w) at (300:2.5) {};
		\node[xshift=8cm,circle, draw=black!100,fill=black!100,thick, inner sep=0pt, minimum size=1.5mm,label=above right:{$x_{1,2}^w$}] (x_{1,2}^w) at (0:1) {};
		\node[xshift=8cm,circle, draw=black!100,fill=black!100,thick, inner sep=0pt, minimum size=1.5mm,label=above right:{$y_{1,2}^w$}] (y_{1,2}^w) at (0:2) {};
		\node[xshift=8cm,circle, draw=black!100,fill=black!100,thick, inner sep=0pt, minimum size=1.5mm,label=below:{$z_{1,2}^w$}] (z_{1,2}^w) at (0:3) {};
		\node[xshift=8cm,circle, draw=black!100,fill=black!100,thick, inner sep=0pt, minimum size=1.5mm,label=right:{$t_2^w$}] (t_2^w) at (60:2.5) {};
		\node[xshift=8cm,circle, draw=black!100,fill=black!100,thick, inner sep=0pt, minimum size=1.5mm,label=right:{$x_{0,2}^w$}] (x_{0,2}^w) at (120:1) {};
		\node[xshift=8cm,circle, draw=black!100,fill=black!100,thick, inner sep=0pt, minimum size=1.5mm,label=above right:{$y_{0,2}^w$}] (y_{0,2}^w) at (120:2) {};
		\node[xshift=8cm,circle, draw=black!100,fill=black!100,thick, inner sep=0pt, minimum size=1.5mm,label=left:{$z_{0,2}^w$}] (z_{0,2}^w) at (120:3) {};
		\draw (t_0^w)--(c^w)--(x_{0,1}^w)--(y_{0,1}^w)--(z_{0,1}^w);
		\draw (t_1^w)--(c^w)--(x_{1,2}^w)--(y_{1,2}^w)--(z_{1,2}^w);
		\draw (t_2^w)--(c^w)--(x_{0,2}^w)--(y_{0,2}^w)--(z_{0,2}^w);
		\draw (t_0^w)--(x_{0,1}^w)--(t_1^w)--(x_{1,2}^w)--(t_2^w)--(x_{0,2}^w)--(t_0^w)--(y_{0,1}^w)--(t_1^w)--(y_{1,2}^w)--(t_2^w)--(y_{0,2}^w)--(t_0^w);
		\draw[dashed] (y_{1,2}^w) -- +(1,-2);
		\draw[dashed] (y_{1,2}^w) -- +(1,2);
		\draw[dashed] (y_{0,1}^w) -- +(0.5,-1);
		\draw[dashed] (t_1^w) -- +(0.25,-0.5);
		\draw[dashed] (y_{0,2}^w) -- +(0.5,1);
		\draw[dashed] (t_2^w) -- +(0.25,0.5);
		
		\draw (t_0^v)--(t_0^w) (y_{0,1}^v)--(y_{0,2}^w) (y_{0,2}^v)--(y_{0,1}^w);
		\end{tikzpicture}}
	\caption{Illustration of the reduction used in the proof of Theorem~\ref{thm:planar}. Here, two vertex gadgets $G_v$, $G_w$ are depicted, with $v$ and $w$ adjacent in $G$. Dashed lines represent potential edges to other vertex-gadgets.}\label{fig:planar-reduction}
\end{figure}
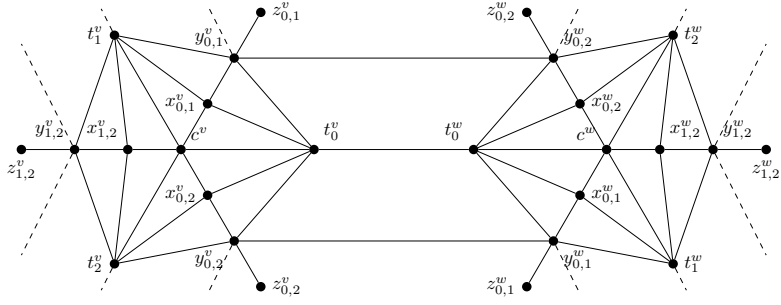

We will show that $G$ has a dominating set of size $k$ if and only if $f(G)$ has a geodetic set of size $3|V(G)|+k$.

Assume first that $G$ has a dominating set $D$ of size~$k$. We construct a geodetic set $S$ of $f(G)$ of size $3|V(G)|+k$ as follows. For each vertex $v$ in $G$, we add the three vertices $z_{i,j}^v$ ($0\leq i<j\leq 2$) of $G_v$ to $S$. If $v$ is in $D$, we also add vertex $c^v$ to $S$.

Let us show that $S$ is indeed a geodetic set. First, we observe that, in any vertex gadget $G_v$ that is part of $f(G)$, the unique shortest path between two distinct vertices $z_{i,j}^v$, $z_{i',j'}^v$ has length~$4$ and goes through vertices $y_{i,j}^v$, $t_{k}^v$ and $y_{i',j'}^v$ (where $\{k\}=\{i,j\}\cap\{i',j'\}$). Thus, it only remains to show that vertices $c^v$ and $x_{i,j}^v$ ($0\leq i<j\leq 2$) belong to some shortest path of vertices of $S$. Assume that $v$ is a vertex of $G$ in $D$. The shortest paths between $c^v$ and $z_{i,j}^v$ have length~$3$ and one of them goes through vertex $x_{i,j}^v$. Thus, all vertices of $G_v$ belong to some shortest path between vertices of $S$. Now, consider a vertex $w$ of $G$ adjacent to $v$ and let $z_{i,j}^w$ be the vertex of $G_w$ that is farthest from $c_v$. The shortest paths between $c^v$ and $z_{i,j}^w$ have length~$6$; one of them goes through vertices $c^w$ and $x_{i,j}^w$; two others go through the two other vertices $x_{i',j'}^w$ and $x_{i'',j''}^w$. Thus, $S$ is a geodetic set.

For the converse, assume we have a geodetic set $S'$ of $f(G)$ of size $3|V(G)|+k$. We will show that $G$ has a dominating set of size~$k$. First of all, observe that all the $3|V(G)|$ vertices of type $z_{i,j}^v$ are necessarily in $S'$, since they have degree $1$. As observed earlier, the shortest paths between those vertices already go through all vertices of type $t_i^v$ and $y_{i,j}^v$. However, no other vertex lies on a shortest path between two such vertices: these shortest paths always go through the boundary $6$-cycle of the vertex-gadgets. Let $S'_0$ be the set of the remaining $k$ vertices of $S'$. These vertices are there to cover the vertices of type $c^v$ and $x_{i,j}^v$. We construct a subset $D'$ of $V(G)$ as follows: $D'$ contains those vertices $v$ of $G$ whose vertex-gadget $G_v$ contains a vertex of $S'_0$. We claim that $D'$ is a dominating set of $G$. Suppose by contradiction that there is a vertex $v$ of $G$ such that neither $G_v$ nor any of $G_w$ (with $w$ adjacent to $v$ in $G$) contains any vertex of $S'_0$. Here also, the shortest paths between vertices of $S$ always go through the boundary $6$-cycle of $G_v$ and thus, they never include vertex $c_v$, a contradiction. Thus, $D'$ is a dominating set of size $k$, and we are done.

\subsection{NP-hardness on line graphs}\label{sec:line-hard}

In this section, we prove that the \textsc{MGS} problem remains NP-hard on line graphs. For a graph $G$ and edges $e,e'\in E(G)$, define $d(e,e')=1$ if $e,e'$ shares a vertex and $d(e,e')=i$ if $e'$ shares a vertex with an edge $e''$ with $d(e,e'')=i-1$. A path between two edges $e,e'$ is defined in the usual way. 

\begin{observation}\label{obs:line-1}
	A path between two edges $e,e'$ of a graph $G$ corresponds to a path between the vertices $e $ and $e'$ in $L(G)$.
\end{observation}

Given a graph $G$, a set $S\subseteq E(G)$ is a \emph{line geodetic set} of $G$ if every edge $e\in E(G)\setminus S$ belongs to some shortest path between some pair of edges $\{e,e'\}\subseteq S$. Observation~\ref{obs:line-1} implies the following. 

\begin{observation}\label{obs:line-2}
	A graph $G$ has a line geodetic set of cardinality $k$ if and only if $L(G)$ has a geodetic set of size $k$.
\end{observation}

We shall show (in Lemma~\ref{lem:edge-cover}) that finding a line geodetic set of a graph with minimum cardinality is NP-hard. Then Observation~\ref{obs:line-2} shall imply that solving the \textsc{MGS} problem on line graphs is NP-hard. For the above purpose we need the following definition. Given a graph $G$, a set $S\subseteq E(G)$ is a \emph{good edge set} if for any edge $e\in E(G)\setminus S$, there are two edges $e',e''\in S$ such that (i) $e$ lies in some shortest path between $e'$ and $e''$, and (ii) $d(e',e'')$ is $2$ or $3$.

\begin{figure}[t]
	\centering
	\scalebox{0.7}{
		\begin{tabular}{cc}
			\includegraphics[page=6]{figures.pdf} & \includegraphics[page=1]{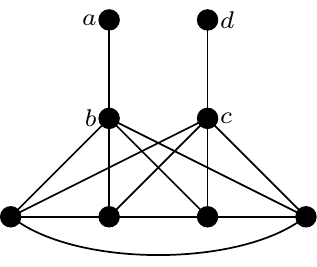} \\
			(a) & (b)
	\end{tabular}}
	\caption{ (a) A triangle-free graph $G$ and (b) the graph $H_G$.}    \label{fig:gadget-2}
\end{figure}


\begin{lemma}\label{lem:goodset}
	Computing a good edge set of a triangle-free graph with minimum cardinality is NP-hard. 
\end{lemma}

\begin{proof}
	We shall reduce the NP-complete \textsc{Edge Dominating set} problem on triangle-free graphs~\cite{yannakakis1980} to the problem of computing a good edge set of a graph with minimum cardinality on triangle-free graphs. Given a graph $G$, a set $S\subseteq E(G)$ is an \emph{edge dominating set} of $G$ if any edge $e\in E(G)\setminus S$ shares a vertex with some edge in $S$. The \textsc{Edge Dominating Set} problem is to compute an edge dominating set of $G$ with minimum cardinality. 
	
	Let $G$ be a triangle-free graph. For each vertex $v\in V(G)$, take a new edge $x_vy_v$. Construct a graph $G^*$ whose vertex set is the union of $V(G)$ and the set $\{x_v,y_v\}_{v\in V(G)}$ and $E(G^*)=E(G) \cup \{vx_v\}_{v\in V(G)} \cup \{x_vy_v\}_{v\in V(G)}$. Notice that $G^*$ is a triangle-free graph and we shall show that $G$ has an edge dominating set of cardinality $k$ if and only if $G^*$ has a good edge set of cardinality $k+n$ where $n=|V(G)|$.

	Let $S$ be an edge dominating set of $G$. For each $v\in G$, let $H_v$ be written as $x_v,y_v,z_v$. Notice that the set $S\cup \{x_vy_v\}_{v\in V(G)}$ forms a good edge set of $G^*$ and has cardinality $k+n$. Let $S'$ be a good edge set of $G^*$ of size at most $k+n$. Notice that for each $v\in V(G)$, $S'$ must contain the edge $x_vy_v$. Hence, the cardinality of the set $S'\cap E(G)$ is at most $k$. Moreover, for each $e\in E(G^*)\cap E(G)$, there is an edge $e'\in S'$ which is at distance $2$ from $e$. As $S'$ is a good edge set of $G^*$, any edge in $E(G)\setminus S'$ shares a vertex with some edge of $S'$. Hence $S'\cap E(G)$ is an edge dominating set of $G$ of cardinality at most $k$.
\end{proof}

For a triangle-free graph $G$, let $H_G$ be the graph with $V(H_G)=V(G)\cup \{a,b,c,d\}$ and $E(H_G)=E(G)\cup \{ab,cd\}\cup  E'$ where $E'=\{bv\}_{v\in V(G)}\cup \{cv\}_{v\in V(G)}$. See Figure~\ref{fig:gadget-2}(a) and Figure~\ref{fig:gadget-2}(b) for an example. We prove the following proposition.

\begin{lemma}\label{lem:line-1}
	For a triangle-free graph $G$, there is a line geodetic set $Q$ of $H_G$ with minimum cardinality such that $Q\cap E'=\emptyset$.
\end{lemma}

\begin{proof}
	For a set $S\subseteq E(H_G)$, an edge $f\in S$ \emph{covers} an edge $e\in E(H_G)$, if there is another edge $f'\in S$ such that $e$ lies in the shortest path between $f$ and $f'$. Notice that the edges $\{ab,cd\}$ lie in any line geodetic set of $H_G$ and all edges in $E'$ are covered by $ab$ and $cd$. First we prove the following claims.

	\begin{claim}
		Let $Q$ be a line geodetic set of $H_G$ and $e\in E'\cap Q$. If $e$ does not cover any edge of $E(G)$, then $Q\setminus \{e\}$ is a line geodetic set of $G^*$. 
	\end{claim}

	\medskip
	\noindent The proof of the above claim follows from the fact that all edges in $E'\cup \{ab,cd\}$ are covered by $ab$ and $cd$.

	\begin{claim}
		Let $Q$ be a line geodetic set of $H_G$ and $e\in E'\cap Q$. There is another edge $e'\in E(G)\setminus Q$ such that $(Q\cup \{e'\})\setminus \{e\}$ is a line geodetic set of $H_G$. 
	\end{claim}
	
	\medskip
	\noindent
	To prove the claim above, first we define the \emph{ecentricity} of an edge $e\in E(H_G)$ to be the maximum shortest path distance between $e$ and any other edge in $E(H_G)$. Notice that the ecentricity of any edge in $E'$ is two and the ecentricity of any edge of $E(G)$ in $H_G$ is at most three. Now remove all edges from $E'\cap Q$ which do not cover any edge of $E(G)$. By Claim 1, the resulting set, say $Q'$, is a line geodetic set of $H_G$. Let $e$ be an edge $Q'\cap E'$ and let $\{f_1,f_2,\ldots,f_k\}\subseteq E(G)\setminus Q'$ be the set of edges covered by $e$. Since the ecentricity of $e$ is two, there must exist $e_1,e_2,\ldots,e_k$ in $Q'$ such that $f_i$ has a common endpoint with both $e$ and $e_i$ for each $i\in\{1,2,\ldots,k\}$. Therefore the distance between $e$ and $e_i$ is two for each $i\in\{1,2,\ldots,k\}$. As $G$ is triangle-free, $e_i\neq e_j$ for any $i,j\in \{1,2,\ldots,k\}$. Choose any edge $f_j\in \{f_1,f_2,\ldots,f_k\}$. Observe that the distance between $f_j$ and $e_i$ is two when $i\neq j$. Therefore, for each $i\in \{i,2,\ldots,j-1,j+1,\ldots k\}$, the edge $f_i$ lies in the shortest path between $f_j$ and $e_i$. Therefore, $(Q'\cup \{f_j\})\setminus \{e\}$ is a line geodetic set of $H_G$.

	\medskip
	Given any line geodetic set $P$ of $H_G$, we can use the arguments used in Claim 1 and Claim 2 repeatedly on $P$ to construct a line geodetic set $Q$ of $H_G$ such that $|Q|\leq |P|$ and $Q\cap E'=\emptyset$. Thus we have the proof.  
\end{proof}

\begin{lemma}\label{lem:edge-cover}
	Computing a line geodetic set of a graph with minimum cardinality is NP-hard.
\end{lemma}

\begin{proof}
	We shall reduce the NP-complete problem of computing a good edge set of a triangle-free graph with minimum cardinality (Lemma~\ref{lem:goodset}). Let $G$ be a triangle-free graph. Construct the graph $H_G$ as stated above (just before Lemma~\ref{lem:line-1}). The set $E'$ is also defined as before. We shall show that a triangle-free graph $G$ has a good edge set of cardinality $k$ if and only if $H_G$ has a line geodetic set of cardinality $k+2$. 
	
	Let $P$ be a good edge set of $G$. Notice that, for each edge $e\in E(G)$, there are two edges $e',e''\in P$ such that $e$ belongs to a shortest path between $e'$ and $e''$ in $H_G$. Also any edge of $E'$ belongs to a shortest path between the edges $ab$ and $cd$ in $H_G$. Hence $P\cup \{ab,cd\}$ is a line geodetic set of $H_G$ with cardinality $k+2$.

	Let $Q$ be a line geodetic set of $H_G$ of size $k+2$. Notice that $\{ab,cd\}\subseteq Q$ and let $Q'=Q\setminus \{ab,cd\}$. Due to Lemma~\ref{lem:line-1}, we can assume that $Q'$ does not contain any edge of $E'$. Let $e$ be an edge in $E(G)\setminus Q'$ and let $e',e''\in Q'$ such that $e$ lies in some shortest path between $e'$ and $e''$ in $H_G$. Since the distance between $e'$ and $e''$ is at most three in $H_G$, it follows that $Q'$ is a good edge set of $G$ with cardinality $k$.  
\end{proof}

\subsection{Inapproximability on graphs with diameter $2$}\label{sec:approx-hard}
%
%
%

Given a graph $G$, a set $S\subseteq V(G)$ is a \emph{2-dominating set} of $G$ if any vertex $w\in V(G)\setminus S$ has at least two neighbours in $S$. The \textsc{2-MDS} problem is to compute a $2$-dominating set of graphs with minimum cardinality. We shall use the following result.

\begin{theorem}[\cite{chlebik2008,dinur2014}]\label{rslt:approx}
	Unless $P=NP$, there is no polynomial time $o(\log n)$-approximation algorithm for the \textsc{2-MDS} problem on triangle-free graphs. 
\end{theorem}

We observe the following.


\begin{lemma}\label{lem:2-dom}
	Let $G$ be a triangle-free graph and $G'$ be the graph obtained by adding an universal vertex $v$ to $G$. A set $S$ of vertices of $G'$ is a geodetic set if and only if $S\setminus \{v\}$ is a $2$-dominating set of $G$.
\end{lemma}
\begin{proof}
	Let $S$ be a geodetic set of $G'$. Observe that for any vertex $u\in V(G) \setminus S$ there must exist vertices $u_1,u_2\in S\setminus \{v\}$ such that $u \in I(u_1,u_2)$ and $u_1u_2\notin E(G)$. Hence, $S$ is a $2$-dominating set of $G$. Conversely, let $S'$ be any $2$-dominating set of $G$. For any two vertex $u\in V(G)\setminus S$ there exist $v,v'\in S$ such that $uv,uv'\in E(G)$. Since $G$ is triangle-free, $v$ and $v'$ are non-adjacent. Hence, $u\in I(v,v')$ and $S'\cup\{v\}$ is a geodetic set of $G'$.
\end{proof}

The proof of Theorem~\ref{thm:reduction-diam2} follows due to Lemma~\ref{lem:2-dom} and Theorem~\ref{rslt:approx}.

\section{Approximation algorithms}\label{sec:approx}

In Section~\ref{sec:approx-general} and Section~\ref{sec:approx-grid} we present approximation algorithms for the \textsc{MGS} problem on general graphs and solid grid graphs, respectively.

\subsection{General graphs}\label{sec:approx-general}


We will reduce the \textsc{Minimum Geodetic Set} problem to the \textsc{Minimum Rainbow Subgraph of Multigraph (MRSM)} problem. A subgraph $H$ of an edge colored multigraph $G$ is \emph{colorful} if $H$ contains at least one edge of each color. Given an edge colored multigraph $G$, the \textsc{MRSM} problem is to find a colorful subgraph of $G$ of minimum cardinality. The following is a consequence of a result due to Tirodkar and Vishwanathan~\cite{tirodkar2017}.

\begin{theorem}[\cite{tirodkar2017}]\label{rslt:rainbow}
	Given an edge colored multigraph $G$, there is a polynomial time $O(\sqrt[3]{n}\log n)$-approximation algorithm to solve the \textsc{MRSM} problem where $n=|V(G)|$.
\end{theorem}

We note that Tirodkar and Vishwanathan~\cite{tirodkar2017} proved the above theorem for simple graphs only, but the proof works for multigraphs as well.

Given a graph $G$ form an edge colored multigraph $H_G$ as follows. The vertex set of $H_G$ is the same as $G$. For each subset $\{u,v,w\}\subseteq V(G)$ such that $u$ lies in some shortest path between $v$ and $w$, add an edge in $H_G$ between $v$ and $w$ having the color $u$. Observe that, $G$ has a geodetic set of cardinality $k$ if and only if $H_G$ has a colorful subgraph with $k$ vertices. The proof of Theorem~\ref{thm:general-approx} follows from Theorem~\ref{rslt:rainbow}.

\subsection{Solid grid graphs}\label{sec:approx-grid}

In this section, we shall give a linear time $3$-approximation algorithm for the \textsc{MGS} problem on solid grid graphs. From now on $G$ shall denote a solid grid graph and $\mathcal{R}$ is a grid embedding of $G$ where every interior face has unit area. 

%
%
%

\begin{figure}[t]
	\centering
	\scalebox{0.7}{
		\begin{tabular}{ccc}
			\includegraphics[page=3]{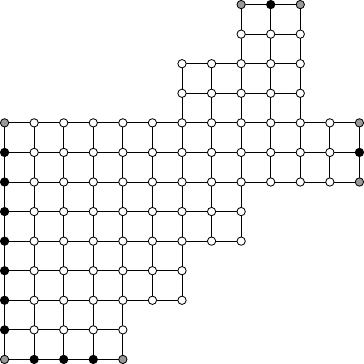} & \includegraphics[page=4]{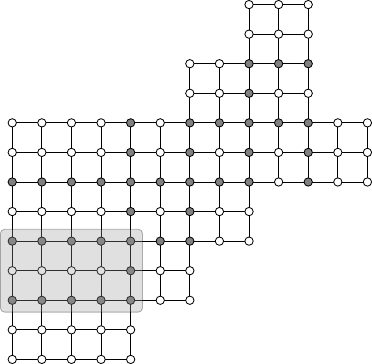} & \includegraphics[page=5,scale=1.5]{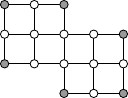} \\
			(a) & (b) & (c)
	\end{tabular}}
	\caption{(a) The black and gray vertices are the vertices of the corner paths. The gray vertices indicate the corner vertices. (b) The gray vertices are vertices of the red path. Vertices in the shaded box form a rectangular block. (c) Example of a solid grid graph whose number of corner vertices is exactly three times the geodetic number.} 
	\label{fig:grid}
\end{figure}

Let $G$ be a solid grid graph. A path $P$ of $G$ is a \emph{corner path} if (i) no vertex of $P$ is a cut vertex, (ii) both end-vertices of $P$ have degree $2$, and (iii) all vertices except the end-vertices of $P$ have degree $3$. See Figure~\ref{fig:grid}(a) for an example. Observe that for a corner path $P$, either the $x$-coordinates of all vertices of $P$ are the same or the $y$-coordinates of all vertices of $P$ are the same. Moreover, all vertices of a corner path lie in the outer face of $G$. The next observation follows from the definition of corner path and the fact that $G$ is a solid grid graph.

\begin{observation}\label{obs:neighbour-path}
	Let $P$ be a corner path of $G$. Consider the set $Q=\{v\in V(G)\colon v\notin V(P), N(v)\cap P\neq \emptyset\}$. Then $Q$ induces a path in $G$. Moreover, if the $x$-coordinates (resp. the $y$-coordinates) of all the vertices of $P$ are the same, then the $x$-coordinates (resp. the $y$-coordinates) of all vertices in $Q$ are the same.
\end{observation}

\noindent We shall use Observation~\ref{obs:neighbour-path} to prove a lower bound on the geodetic number of $G$ in terms of the number of corner paths of $G$.

\begin{lemma}\label{lem:cardinality}
	Any geodetic set of $G$ contains at least one vertex from each corner path.
\end{lemma}

\begin{proof}
	\sloppy Without loss of generality, we assume the $x$-coordinates of all vertices of $P$ are the same. By Observation~\ref{obs:neighbour-path}, the set $\{v\in V(G)\colon v\notin V(P), N(v)\cap P\neq \emptyset\}$ induces a path $Q$ and the $x$-coordinates of all vertices in $Q$ are the same. Now consider any two vertices $a,b\in V(G)\setminus V(P)$ and with a path $P'$ between $a$ and $b$ that contains one of the end-vertices, say $u$, of $P$. Observe that $P'$ can be expressed as $P'=a~c_1~c_2\ldots c_t~d~f_1~f_2\ldots f_{t'}~u~g~h_1~h_2~\ldots h_{t''}~b$ such that $\{d,g\}\subseteq V(Q)$ and $\{f_1,f_2,\ldots,f_{t'}\}\subseteq V(P)$. Then there is a path $P''=a~c_1~c_2\ldots c_t~d~f'_2\ldots f'_{t'}~g~h_1~h_2~\ldots h_{t''}~b$ where for $2\leq i\leq t'$, $f'_i$ is the vertex in $Q$ which is adjacent to $f_i$ in $G$. Observe that the length of $P''$ is strictly less than that of $P'$. Therefore $u\notin I(a,b)$ whenever $a,b\in V(G)\setminus V(P)$. Hence any geodetic set of $G$ contains at least one vertex from $P$.
\end{proof}

Any geodetic set of $G$ contains all vertices of degree $1$. Inspired by the above fact and Lemma~\ref{lem:cardinality}, we define the term \emph{corner vertex} as follows. A vertex $v$ of $G$ is a \emph{corner vertex} if $v$ has degree $1$ or $v$ is an end-vertex of some corner path. See Figure~\ref{fig:grid}(a) for an example. Observe that two corner paths may have at most one corner vertex in common. Moreover, a corner vertex cannot be in three corner paths. Therefore it follows that the cardinality of the set of corner vertices is at most $3\cdot g(G)$. 

\begin{remark}
	Note that there are solid grid graphs whose number of corner vertices is exactly three times the geodetic number. See Figure~\ref{fig:grid}(c) for one such example.
\end{remark}

Now we prove that the set of all corner vertices of $G$ is indeed a geodetic set of $G$. We shall use the following proposition of Ekim and Erey~\cite{ekim2014}.

\begin{theorem}[\cite{ekim2014}]\label{rslt:cut-vertex-geod}
	Let $F$ be a graph and $F_1,\ldots,F_k$ its biconnected components. Let $C$ be the set of cut vertices of $G$. If $X_i\subseteq V(F_i)$ is a minimum set such that $X_i\cup (V(F_i)\cap C)$ is a minimum geodetic set of $F_i$ then $\cup_{i=1}^{k} X_i$ is a minimum godetic set of $F$.
\end{theorem}

The next observation follows from Theorem~\ref{rslt:cut-vertex-geod}.

\begin{observation}\label{obs:block}
	Let $C(G)$ be the set of corner vertices of $G$ and $S$ be the set of cut vertices of $G$. Let $\{H_1,H_2,\ldots,H_t\}$ be the set of biconnected components of $G$. The set $C(G)$ is a geodetic set of $G$ if and only if $(C(G)\cap V(H_i))\cup (S\cap V(H_i))$ is a geodetic set of $H_i$ for all $1\leq i\leq t$.
\end{observation}

From now on, $C(G)$ is the set of corner vertices of $G$ and $H_1,H_2,\ldots,H_t$ are the biconnected components of $G$. Due to Theorem~\ref{rslt:cut-vertex-geod} and Observation~\ref{obs:block}, it is enough to show that for each $1\leq i \leq t$, the set $(C(G)\cap V(H_i))\cup (S\cap V(H_i))$ is a geodetic set of $H_i$. First, we introduce some more notations and definitions below. 

\medskip

Let $H$ be a biconnected component of $G$. Recall that each vertex of $H$ is a pair of integers and each edge is a line segment with unit length. 
An edge $e\in E(H)$ is an \emph{interior} edge if all interior points of $e$ lie in an interior face of $H$. For a vertex $v\in V(H)$, let $P_v$ denote the maximal path such that all edges of $P_v$ are interior edges and each vertex in $P_v$ has the same $x$-coordinate as $v$. Similarly, let $P'_v$ denote the maximal path such that all edges of $P_v$ are interior edges and each vertex in $P'_v$ has the same $y$-coordinate as $v$. A path $P$ of $H$ is a \emph{red path} if (i) there exists a $v\in V(H)$ such that $P\in \{P_v,P'_v\}$ and (ii) at least one end-vertex of $P$ is a cut-vertex or a vertex of degree $4$. A vertex $v$ of $H$ is \emph{red} if $v$ lies on some red path. See Figure~\ref{fig:grid}(b) for an example. 

\begin{definition}
	A subgraph $F$ of $H$ is a rectangular block if $F$ satisfies the following properties.
	
	\begin{enumerate}
		\item For any two vertices $(a_1,b_1),(a_2,b_2)$ of $F$, we have that any pair $(a_3,b_3)$ with $a_1\leq a_3\leq a_2$ and $b_1\leq b_3\leq b_2$ is a vertex of $F$.
		
		\item\label{it:prop-2} Let $a,a'$ be the maximum and minimum $x$-coordinates of the vertices in $F$. The $x$-coordinate of any red vertex of $F$ must be equal to $a$ or $a'$. Similarly, let $b,b'$ be the maximum and minimum $y$-coordinates of the vertices in $F$. The $y$-coordinate of any red vertex of $F$ must be equal to $b$ or $b'$. 
		
	\end{enumerate}
\end{definition}

Observe that $H$ can be decomposed into rectangular blocks such that each non-red vertex belongs to exactly one rectangular block. See Figure~\ref{fig:grid}(b) for an example. Let $B_1,B_2,\ldots,B_k$ be a decomposition of $H$ into rectangular blocks. Recall that $C(G)$ is the set of corner vertices of $G$ and $S$ is the set of cut vertices of $G$. We have the following lemma.

\begin{lemma}\label{lem:block-cover}
	For each $1\leq i\leq k$, there are two vertices $x_i,y_i\in (C(G)\cap V(H))\cup (S\cap V(H))$ such that $V(B_i)\subseteq I(x_i,y_i)$.
\end{lemma}

\begin{proof}
Let $X\in \{B_1,B_2,\ldots,B_k\}$ be an arbitrary rectangular block. A vertex $v$ of $X$ is a \emph{northern} vertex if the $y$-coordinate of $v$ is maximum among all vertices of $X$. Analogously, \emph{western} vertices, \emph{eastern} vertices and \emph{southern} vertices are defined. A vertex of $X$ is a boundary vertex if it is either northern, western, southern or an eastern vertex of $X$. Let $nw(X)$ be the vertex of $X$ which is both a northern vertex and a western vertex. Similarly, $ne(X)$ denotes the vertex which is both northern vertex and eastern vertex, $sw(X)$ denotes the vertex of $X$ which is both southern and western vertex and $se(X)$ denotes the vertex of $X$ which is both southern and eastern vertex. 

%

First we prove the lemma assuming that all boundary vertices of $X$ are red vertices. Let $a$ (resp. $b$) denote the vertex with minimum $y$-coordinate such that $P_a$ (resp. $P_b$) contains $sw(X)$ (resp. $se(X)$). Similarly, let $c$ (resp. $d$) denote the vertex with maximum $y$-coordinate such that $P_c$ (resp. $P_d$) contains $nw(X)$ (resp. $ne(X)$). Let $a'$ (resp. $c'$) denote the vertex with minimum $x$-coordinate such that $P'_{a'}$ (resp. $P'_{b'}$) contains $sw(X)$ (resp. $nw(X)$). Let $b'$ (resp. $d'$) denote the vertex with maximum $x$-coordinate such that $P'_{b'}$ (resp. $P'_{d'}$) contains $se(X)$ (resp. $ne(X)$). Observe that the vertices $a',a,b,b',d',d,c,c'$ lie on the exterior face of the embedding. 

For two vertices $i,j\in\{a',a,b,b',d',d,c,c'\}$, let $P_{ij}$ denote the path between $i,j$ that can be obtained by traversing the exterior face of the embedding in the counter-clockwise direction starting from $i$. Observe that, if both $P_{a'a}$ and $P_{d'd}$ (resp. $P_{bb'}$ and $P_{cc'}$) contain a corner or cut vertex each, say $f,f'$, then $\left\{sw(X),ne(X)\right\}\subseteq I(f,f')$ (resp. $\left\{nw(X),se(X)\right\}\subseteq I(f,f')$) and therefore $V(X)\subseteq I(f,f')$. Now consider the case when at least one of the paths in $\{P_{a'a},P_{d'd}\}$ does not contain any corner vertex or cut vertex and when at least one of the paths in $\{P_{b'b},P_{cc'}\}$ does not contain any corner vertex or cut vertex. Due to symmetry of rotation and reflection on grids, without loss of generality we can assume that both $P_{a'a}$ and $P_{bb'}$ have no corner vertex or cut vertex. Observe that in this case there must be a corner vertex $f$ in $P_{ab}$ whose $x$-coordinate is the same as that of $b$ and therefore of $se(X)$. If $P_{cc'}$ contains a corner vertex $f'$, then $\left\{nw(X),se(X)\right\}\subseteq I(f,f')$) and therefore $V(X)\subseteq I(f,f')$. Otherwise, there must be a corner vertex $f'$ in $P_{c'a'}$ whose $y$-coordinate is the same as that of $c'$ and therefore of $nw(X)$. Hence we have $\left\{nw(X),se(X)\right\}\subseteq I(f,f')$ and therefore $V(X)\subseteq I(f,f')$ in this case also.

Now we consider the case when there are some non-red boundary vertices of $X$. Let $v$ be a non-red vertex of $X$. Without loss of generality, we can assume that $v$ is a western vertex of $X$. Now we redefine the vertices $a,a',b,b',c,c',d,d'$ as follows. Let $a'=sw(X)$, $c'=nw(X)$ and $a$ (resp. $b$) be the vertex with minimum $y$-coordinate such that there is a path from $a$ to $sw(X)$ (resp. from $b$ to $se(X)$) containing vertices with the same $x$-coordinate as that of $sw(X)$ (resp. $se(X)$). Similarly, let $c$ (resp. $d$) be the vertex with maximum $y$-coordinate such that there is a path from $c$ to $nw(X)$ (resp. from $d$ to $ne(X)$) containing vertices with the same $x$-coordinate as that of $nw(X)$ (resp. $ne(X)$). Finally, let $d'$ (resp. $b'$) be the vertex with maximum $x$-coordinate such that there is a path from $d'$ to $ne(X)$ (resp. from $b'$ to $se(X)$) containing vertices with the same $y$-coordinate as that of $ne(X)$ (resp. $se(X)$). Using similar arguments on the paths $P_{ij}$ with $i,j\in\{a',a,b,b',d',d,c,c'\}$ as before, we can show that there exists corner vertices $f,f'$ such that $V(X)\subseteq I(f,f')$. Thus we have the proof.
\end{proof}

By Observation~\ref{obs:block} and Lemma~\ref{lem:block-cover}, $C(G)$ is a geodetic set of $G$. 

\medskip \noindent\textbf{Time complexity:} If the grid embedding of $G$ is given as part of the input, then the set of corner vertices can be computed in $O(|V(G)|)$ time by simply traversing the exterior face of the embedding. Otherwise, the set of corner vertices can be computed in $O(|V(G)|)$ time as follows (we shall only describe the procedure to find corner vertices of degree two as the other case is trivial). Let $H$ be a biconnected component of $G$, $v$ be a vertex of $H$ having degree $2$ and $u_0,x_0$ be its neighbours. If both $u_0$ and $x_0$ have degree $4$, then $v$ is not a corner vertex. Moreover, if at least one of $u_0$ and $x_0$ have degree $2$ then $v$ is a corner vertex. Otherwise, apply the following procedure. Assume $u_0$ has degree $3$ and denote $v$ as $u_{-1}$ for technical reasons. Set $i=0$. As $H$ is a biconnected solid grid graph, $u_i$ and $x_i$ must have exactly one common neighbour which is different from $u_{i-1}$. Denote this vertex as $x_{i+1}$. Let $u_{i+1}$ be the neighbour of $u_i$ different from both $x_{i+1}$ and $u_{i-1}$. If $deg_H(u_{i+1})=4$ or $u_{i+1}$ is a cut vertex in $G$ then terminate. If $deg_G(u_{i+1})=2$ then $v$ is a corner vertex. Otherwise, set $i=i+1$ and repeat the above steps. Observe that, when the above procedure terminates either we know that $v$ is a corner vertex or there is no corner path that contains both $u_0$ and $v$. Now swapping roles of $u_0$ and $x_0$ in the above procedure, we can decide if $v$ is a corner vertex. We can find all the corner vertices of $H$ by applying the above procedure to all vertices of degree $2$ of $H$. Similarly by applying the above procedure to all the biconnected components of $G$, we can find all corner vertices. Notice that, the total running time of the algorithm remains linear in the number of vertices of $G$.

\medskip

This completes the proof of Theorem~\ref{thm:grid-approx}.

\section{Conclusion}\label{sec:conclude}

In this paper, we studied the computational complexity of the \textsc{MGS} problem in various graph classes. We proved that the \textsc{MGS} problem remains NP-hard on planar graphs and line graphs. This motivates the following question.

\begin{question}
	Are there constant factor approximation algorithms for the \textsc{MGS} problem on planar graphs and line graphs?
\end{question} 

We gave an $O\left(\sqrt[3]{n}\log n\right)$-approximation algorithm for the \textsc{MGS} problem on general graphs and proved that unless P=NP, there is no polynomial time $o(\log n)$-approximation algorithm for the \textsc{MGS} problem even on graphs with diameter $2$. The following is a natural question in this direction.

\begin{question}
	Is there a $O(\log n)$-approximation algorithm for the \textsc{MGS} problem on general graphs ?
\end{question}




\bibliography{reference}

\end{document}